\documentclass[submission,copyright,creativecommons]{eptcs}
\providecommand{\event}{Gandalf 2018} 

\usepackage{underscore}           
\usepackage[english]{babel} 
\usepackage[utf8]{inputenc}
\usepackage{amsfonts}
\usepackage{amsthm}
\usepackage{amssymb,amsmath}
\usepackage{wasysym}
\usepackage{dsfont}
\usepackage{graphics}
\usepackage{graphicx}
\usepackage{wrapfig, blindtext}
\usepackage{subcaption}
\usepackage{epstopdf}
\usepackage{algorithm} 
\usepackage{algorithmic}
\usepackage[x11 names, rgb]{xcolor}
\usepackage{tikz}
\usetikzlibrary{decorations,arrows,shapes,automata}
\usepackage{multicol}
\usepackage{hyperref}
\hypersetup{ colorlinks=true}
\usepackage{color}
\usepackage{framed}
\usepackage{lmodern}

\newif{\ifMarginalComments}
\MarginalCommentstrue

\newcounter{ncomm}

\newtheorem{theorem}{Theorem}
\newtheorem{lemma}[theorem]{Lemma}
\newtheorem{proposition}[theorem]{Proposition}

\newtheorem{definition}[theorem]{Definition}
\newtheorem{example}[theorem]{Example}
\newtheorem{problem}[theorem]{Problem}
\newtheorem{remark}[theorem]{Remark}

\newcommand{\calG}{\mathcal{G}}
\newcommand{\Gcal}{\mathcal{G}}

\newcommand{\calP}{\mathcal{P}}

\newcommand{\bbN}{\mathbb{N}}

\newcommand{\bbP}{\mathbb{P}}

\def\distr{\mathsf{Distr}}
\def\e{\mathsf{E}}
\def\b{\mathsf{B}}
\def\a{\mathsf{A}}
\def\m{\mathsf{M}}
\def\prob#1{\bbP_{#1}}
\def\bm{\mathsf{BM}}
\def\set#1{\{#1\}}
\def\wattr{\mathsf{Attr}}
\def\away{\mathsf{Away}}
\def\bsigma{\bar{\sigma}}
\def\asigma{\vec{\sigma}}
\def\mtau{\bar{\tau}}
\def\parity{\mathsf{Par}}
\def\plays{\mathsf{Plays}}
\def\first{\mathsf{First}}
\def\last{\mathsf{Last}}
\def\paths{\mathsf{Paths}}
\def\supp{\mathsf{Supp}}
\def\s{\mathsf{Safe}}
\def\us{\overline{\s}}

\title{Banach-Mazur Parity Games\\ and\\ Almost-sure Winning Strategies}
\author{Youssouf Oualhadj
\institute{LACL, U-Pec\\ Paris, France}
\email{youssouf.oualhadj@lacl.fr}
\and
Léo Tible
\institute{ENS Paris Saclay\\ Paris, France}
\email{ltible@ens-paris-saclay.fr}
\and 
Daniele Varacca
\institute{LACL, U-Pec\\ Paris, France}
\email{daniele.varacca@lacl.fr}
}
\def\titlerunning{Banach-Mazur Parity Games}
\def\authorrunning{Y. Oualhadj, L. Tible \& D. Varacca}
\begin{document}
\maketitle

\begin{abstract}
  Two-player stochastic games (sometimes referred to as
  $2\frac{1}{2}$-player games) are games with two players and a
  randomised entity called ``nature''. A natural question to ask in
  this framework is the existence of strategies that ensure that an event
  happens with probability 1 (almost-sure strategies). 
  In the case of Markov decision processes ($1\frac{1}{2}$-player games), when the event of
  interest is given as a parity condition, we can replace the "nature" by two more players that play
  according to the rules of what is known as \emph{Banach-Mazur game}~\cite{ACV10}.
  In this paper we continue this research program by extending the
  above result to two-player stochastic parity games. As in the
  paper~\cite{ACV10}, the basic idea is that, under the correct
  hypothesis, we can replace the randomised player with two players playing a
  Banach-Mazur game. This requires a few technical observations, and a
  non trivial proof, that this paper sets out to do.  
\end{abstract}

\section{Introduction}
In the fields of control and design of reactive systems, one often
faces the problem of verifying whether a system, which is interacting
with its environment, has the desired behaviour or not. The
mathematical interpretation of this problem, also known as Church
synthesis problem, is usually modelled by a game played on
graphs. Such a game involves two players, the first one shall be
called \emph{Eve} represents the controller and the second one shall
be called \emph{Adam} and represents the environment\footnote{one
  usually supposes the environment to be antagonist}. The behaviour we
want the controller to ensure, usually called the \emph{objective}, is
given by a set of infinite paths over a graph. In order to check
whether Eve can ensure the objective, the graph is partitioned into two
sets of states; Eve's state and Adam's state. When the play is in Eve's
states, she chooses the next state of the play, and when it is in Adam's
states, he chooses the successor. Therefore, a play generates an infinite
path and Eve wins the play if it is in the objective. Finally, the
problem is to synthesise (compute) a strategy for Eve to ensure
that the generated play is always in the objective.

Two-player stochastic games are a generalisation of the former model in which
the graph is partitioned into three sets of states; Eve's,
Adam's, and \emph{Nature}'s. The role of Nature is to add an
element of surprise. When the play is in a Nature's state, the
successor is chosen according to a \emph{coin flip}. In this case the outcome of a
play is no more a unique path but a set of paths. The synthesis
problem becomes then whether there exists a strategy for Eve such that the
measure of the set of generated paths is larger than a given
threshold. In~\cite{GH10}, it is shown that under the appropriate
assumptions, the core analysis amounts to deciding whether there exists
a strategy such that the measure of the generated paths is 1. Call such a
strategy an \emph{almost-surely winning} strategy.

From a system design point of view, as already mentioned in~\cite{VV12},
the randomisation adds some kind of fairness in the general behaviour
of the system. For instance, consider a game
where Eve wins if she repeats a self loop infinitely. If even a little randomisation is added so that Eve is not sure that she can stay in the loop, she cannot follow this strategy any longer. Thus, one can
say that winning in the framework of Stochastic games is somehow more
realistic as one cannot count on a contrived behaviour to win.

A different approach for implementing fairness is by means of
topology. The main idea is that a strategy is winning if the set paths
it induces is \emph{topologically large}. In some rather general cases,
sets with measure 1 coincide with those topologically large sets. In
particular, Staiger~\cite{Staiger97} and Varacca and Volzer~\cite{VV12} in a
separate work have shown that on finite Markov chains,
$\omega$-regular sets of infinite words have measure one if and
only if they have the topological property of being \emph{co-meager}.
This property can be equivalently characterised by the notion of
Banach-Mazur games~\cite{Ox}, where two players alternately play
finite sequences of paths on the graph structure of the Markov
chain. Brihaye et. al. have shown that this result extends to countable intersection of
$\omega$-regular sets~\cite{BHM15}.

To further the game-theoretic intuition of~\cite{VV12}, Asarin
et. al.~\cite{ACV10} have shown an equivalence between finite Markov
decision processes (one-player stochastic games) on $\omega$-regular
objectives, and a new notion of three-player games, where the players
of the Banach-Mazur game alternately try to help the controller
satisfy its objective, or spoil it. What that paper reveals is that,
while the basic idea of replacing a probabilistic notion with a
topological one is sound and intuitive, some technicalities have to be
spelled out correctly in order for the framework to work.

Our contributions are as follows. We introduce the notion of two-player
$\bm$ parity games which is a generalisation of Banach-Mazur games
(c.f. Section~\ref{sec:eabm}). We show that these games are
positionally determined (c.f. Theorem~\ref{thm:PosDet}) for parity
objectives using an approach à la Zielonka. We also show that a slight
change in the order of quantifier of the definition of the game is enough to lose determinacy
(c.f. Example~\ref{ex:determinacy}).   Finally, we draw a link
between two-player stochastic parity games and two-player $\bm$ parity
games. In particular, we show that there exists an almost-surely
winning strategy in a two-player stochastic parity game if and only if
there exists a winning strategy in a well chosen two-player $\bm$ game
(c.f. Theorem~\ref{thm:trsfrt}), this last result subsumes the one
of~\cite{ACV10}, although the proof there, being applied to a simpler case, is simpler.


\section{Preliminaries}
\subsection{Stochastic games}
For a finite set $S$ we denote $\distr(S)$ the set of all discrete probability
distributions over $S$, that is the set of functions $d: S\to [0,1]$
such that $\sum_{s\in S}d(s) = 1$. For a distribution $d$ in
$\distr(S)$, we denote by $\supp(d)$ the set $\set{s\in S \mid d(s) > 0}$.

Given a graph $(S,E)$, and an element $s\in S$, we denote $\paths(S,E,s)$ as the set
of infinite sequences $w\in S^\omega$ such that $w(0)=s$, and for any $i$, $(w(i),w(i+1))\in E$.

\paragraph{Games and plays}
A \emph{stochastic game} is a tuple $\calG=(S,(S_{\e},S_{\a}),A,\calP)$
where $S$ is a finite set of states, $(S_{\e},S_{\a})$ is a partition
of $S$ such that $S_{\e}$ is the set of states controlled by $\e$ (Eve)
and $S_{\a}$ is the set of states controlled by $\a$ (Adam), $A$ is a
finite nonempty set of actions, and $\calP:S\times A\to \distr(S)$ is
a total transition function.

A \emph{play} is an infinite sequence $s_0a_0s_1a_1\cdots\in (SA)^\omega$.
We denote by $S_n$ the random variable with
values in $S$ that maps each play to its $n$th state and by $A_n$ the
random variable with values in $A$ that maps each play to its $n$th
action. Formally $S_n(s_0a_0s_1a_1\cdots)=s_n$, and
$A_n(s_0a_0s_1a_1\cdots)=a_n$.

\paragraph{Strategies and Measures}
A \emph{strategy} for $\e$ is a function that tells her what is the next action
to play, given a partial play of the game. Formally it is a function $\sigma:(SA)^*S_{\e}\to A$.  We define
strategies for $\a$ as $\tau:(SA)^*S_{\a}\to A$.

Once a pair of strategies is chosen $(\sigma,\tau)$ and an initial
state $s$ is fixed, we associate the probability measure
$\prob{s}^{\sigma,\tau}$ over $S^\omega$ as the only measure over the Borel sets of
$S^\omega$ such that:

\begin{align*}
  &\prob{s}^{\sigma,\tau}(S_0 = s)=1\enspace,\\
  &\prob{s}^{\sigma,\tau}(S_{n+1}= s \mid S_n = s_n) = 
    \begin{cases}
      \calP(s_n,\sigma(s_0a_0\cdots s_n))(s_{n+1})\text{ if } s_n \in S_{\e}\enspace.\\
      \calP(s_n,\tau(s_0a_0\cdots s_n))(s_{n+1})\text{ if } s_n \in
      S_{\a}\enspace.
    \end{cases}
\end{align*}
The existence and uniqueness of such a measure is a consequence of Carathéodory's extension theorem.

\paragraph{Objectives}
An objective is a measurable subset of plays $\Phi\subseteq S^\omega$.
We say that $\e$ wins almost-surely from a state
$s$ if she has a strategy $\sigma$ such that for every strategy
$\tau,~\prob{s}^{\sigma,\tau}(\Phi)=1$.

\subsection{Banach-Mazur Games}

The notion of Banach-Mazur game \cite{Ox} can be presented with different levels of generality. Here we choose to present it in a form that is most suitable to our needs.

Let $T$ be a set, and $X\subseteq T^\omega$ a set of infinite words, and $\Phi\subseteq X$ an objective. The \emph{Banach-Mazur game} on $X$ with objective $\Phi$ is played as follows. There are two players, that we can call Banach (le ``bon'') and Mazur (le ``méchant''). Mazur begins by playing a finite prefix $w_0$ of some word in $X$.
Then Banach extends $w_0$ with another finite prefix $w_1$ of some word in $X$. The play continues, generating an infinite sequence $w_0<w_1<w_2\ldots$ If the limit of this sequence belongs to the objective, then Banach wins, otherwise Mazur wins. It was proven by (the real) Banach and Mazur that Banach wins if and only if the objective has the topological property of being \emph{co-meager} in the Cantor topology induced on $X$ by $T^\omega$ \cite{Ox}.

A special case is when there is a graph structure $(S,E)$, an initial
element $s\in S$ and $X$ is $\paths(S,E,s)$. In this
cas we talk about the Banach-Mazur game on a graph.

\paragraph{Banach-Mazur games and Markov chains}

A Markov chain on a set of states $S$ is given by a function 
$\calP:S\to \distr(S)$. It induces a graph $(S,E_{\calP})$, where
$(s,s')\in E_{\calP}$ if and only $s'\in\supp(\calP(s))$. A Markov
chain can be seen as a stochastic game where the players always have exactly one available choice.
Therefore, similarly to what we have described
above, given an initial state $s$, a Markov chain generates a Borel
probability measure on $\paths(S,E_{\calP},s)$ It is well known that $\omega$-regular sets are measurable~\cite{BK08}.
Varacca and Völzer~\cite{VV12} (see also~\cite{Staiger97}) have shown the following result:
\begin{theorem}
\label{thm:MCBM}
Let $\calP$ be a Markov chain on a finite set $S$. Let $s$ be an initial state, and let $\Phi$ be an $\omega$-regular subsets of $\paths(S,E_\calP,s)$. Then $\Phi$ has measure 1 under $\calP$ if and only if Banach wins the Banach-Mazur game on $(S,E_\calP,s)$ with objective $\Phi$.
\end{theorem}

\subsection{Parity games}

\begin{definition}
  Let $\calG$ be a game, and $\chi:S\to C$ be a priority function where
  $C \subseteq \bbN$. The parity objective $\parity$ is given by the
  following set
  \[
    \parity = \set{ s_0s_1s_2\cdots \in S^\omega \mid
      \limsup(\chi(s_0)\chi(s_1)\chi(s_2)\cdots) \text{ is even} } \enspace.
  \]
\end{definition}
In the sequel, we call games equipped with parity 
objectives \emph{parity games}.

In the setting of stochastic parity games, a natural question to ask
is whether there exists a strategy that ensures the parity objective
with probability 1? Formally,

\begin{problem}[Almost-sure Parity]
  Given a stochastic parity game with initial state $s$, compute a
  strategy $\sigma$ if it exists such that
  \[
    \forall \tau,~
    \prob{s}^{\sigma,\tau}(\parity) = 1 \enspace.
  \]
\end{problem}

A strategy for a player is \emph{positional} (or
\emph{memoryless}) if the choices depend only on the last state of the
current play. Formally, a strategy is positional if it  defines a mapping $\sigma : S \to A$.
An instrumental result in our subsequent development is the following
theorem due to~\cite{Zielonka04, CJH03}.

\begin{theorem}[Positional determinacy]
  \label{thm:stocPos}
  For finite stochastic games with parity objectives, from every state
  $s$, either $\e$ has a positional almost-surely winning strategy, or $\a$ has a positional positively winning strategy.
\end{theorem}



\section{Two-player $\bm$ games}
\label{sec:eabm}
In this section we present the proposal for a notion of four-player games, that correspond to two-player stochastic games, in a formal sense that we show below. 

\subsection{The game and the plays}
\paragraph{Arenas}
A \emph{two-player $\bm$ game} is a tuple $\calG=(S,(S_{\e},S_{\a}),A,\bm)$
where $S$ is a finite set of states, $(S_{\e},S_{\a})$ is a partition
of $S$, $A$ is a
finite nonempty set of actions, and $\bm \subseteq S\times A\times S$ is
a transition relation. We will assume, for simplicity, that for every $s\in S$ and every $a\in A$ there exists at least a state $s'$ for which $(s,a,s')\in\bm$.

The game is played by four players: the \emph{Arena players} $\e$
(Eve) and $\a$ (Adam), and the \emph{Nature players} $\b$ (Banach) and $\m$ (Mazur).

$S_{\e}$ is the set of states controlled by Eve while
$S_{\a}$ is the set of states controlled by Adam.
When it is her turn to play, in a state $s$, $\e$ chooses an action
$a$.  Similarly for $\a$. 
Once an Arena player has made his or her choice, it is the turn of
the Nature players to  choose a next state according to the transition
relation. 
They can also choose to pass their turn. 
Formally, after $\e$ (or $\a$) has played in state $s$ choosing an
action $a$, $\b$ chooses a state $s'$ such that $(s,a,s')\in\bm$. He
can also choose a special action $\bot$ that means that he is passing
the turn. In such a case it is immediately $\m$ that has to make a
choice. (And dually exchanging the roles of the Nature players). 
$\b$ and $\m$ must pass their turn at some point during a play. At the
beginning of the play, it is always $\m$'s turn.

Next we formalise these intuitions.

\paragraph{Winning plays}
A (legal) \emph{play} is an infinite word in
$S((AS)^*A\bot(AS))^\omega$. 
It must contain infinitely many occurrences of $\bot$, and no adjacent
occurrences of $\bot$, We denote the set of plays in a game $\calG$ by
$\plays(\calG)$. 
A \emph{flattening} of a play is obtained by
projecting the play on $S^\omega$, forgetting the occurrences of $A$
and of $\bot$. The flattening of a play is always an infinite word.
An \emph{objective} is a set of infinite words
$\Phi\subseteq S^\omega$.

\begin{definition}
A play $w$ on a game $\calG$ with objective $\Phi$ is \emph{winning} for $\e$ and $\b$ if 
its flattening belongs to the objective. Otherwise the play is winning
for $\a$ and $\m$.
When clear from the context, we will say that the
  play is winning for $\e$ (resp. $\a$).
\end{definition}

\subsection{Turn-based Strategies}
In order to define strategies, we first need to figure out which
information the players are allowed to have. Nature players can try to
help their allied (Banach helps Eve, while Mazur helps Adam) but Arena
players must not know against which one of the Nature players they are
playing at any given moment. We will explain why at the end of the
section (c.f. Example~\ref{ex:determinacy}).
Therefore, Arena players are not allowed to see the occurrences of
$\bot$ in choosing what move to make, and strategies of Arena players
are defined precisely as in stochastic games. 

Strategies for $\b$ and $\m$ are mappings from partial plays ending in
$A\cup \bot$, choosing either an element in $S$ or $\bot$.
Formally a strategy  for $\b$ is a function $\bsigma:S((AS)^*A\bot(AS))^*A(\bot?)\to S\cup\bot$ such that:
\begin{enumerate}
\item if $\bsigma(wsa)=s'$ then $(s,a,s')\in\bm$
\item if $\bsigma(wsa\bot)=s'$ then $(s,a,s')\in\bm$
\item $\bsigma(wsa\bot)\neq\bot$
\item there cannot be an infinite sequence $ws_0a_0s_1a_1\ldots$ such
  that $\forall i\ge 0,~\bsigma(ws_0a_0\ldots s_ia_i)=s_{i+1}$
\end{enumerate} 
Similarly for $\mtau$.
Condition (3) insures that strategies always produce infinite plays. It corresponds the property called
\emph{progressiveness} by Varacca and Völzer~\cite{VV12}.
Condition (4) insures that a Nature player must eventually play $\bot$.

Once all the players have chosen their strategies, a play is
obtained. 
In order to give the formal definition, we need the
following notation:
\begin{itemize}
\item given a finite sequence $w$, we denote by $\last(w)$ its last
  element;
\item given a finite sequence $w$ possibly containing occurrences of
  $\bot$, we denote by $\pi_\bot(w)$ the sequence obtained by
  eliminating these occurrences.
\end{itemize}

Formally, given a play $\tilde{w}$ in $\plays(\calG)$ we say that it respects
the strategies $\sigma$ for $\e$, $\tau$ of $\a$, $\bsigma$ for $\b$,
$\mtau$ for $\m$ if for every finite prefix $w$ of $\tilde{w}$:

 \begin{enumerate}
 \item if $\last(w)\in S_{\e}$ and $\sigma(\pi_\bot(w))=a$ then $wa$
   is a prefix of $\tilde{w}$
 \item if $\last(w)\in S_{\a}$ and $\tau(\pi_\bot(w))=a$ then $wa$ is
   a prefix of $\tilde{w}$
 \item if $\last(w)\in A\cup \bot$, $w$ contains an odd number of
   occurrences of $\bot$, and $\bsigma(w)=x$, then $wx$ is a prefix of
   $\tilde{w}$.
 \item if $\last(w)\in A\cup \bot$, $w$ contains an even number of
   occurrences of $\bot$, and $\mtau(w)=x$, then $wx$ is a prefix of
   $\tilde{w}$.
 \end{enumerate}

 \begin{remark}
   Given four strategies for the four players, and an initial state
   $s$, there is a unique play of the game $\calG$ beginning in
   $s$ that respects them that we call the induced play.
 \end{remark}

 \begin{definition}
   We say that $\e$ and $\b$ \emph{win structurally} from a state
   $s$, if there exists a strategy $\sigma$ for $\e$ such that for any
   strategy $\tau$ of $\a$ there exists a strategy $\bsigma$ for $\b$
   such that for any strategy $\mtau$ for $\m$ the induced play
   starting from $s$ is winning for $\e$ and $\b$.
 \end{definition}

\subsection{Global Strategies} 
We have presented this turn based way
of playing the game, as it is intuitive. However it is hard to work
with it, as the definition of strategy for the Nature Players is quite
involved. We propose here an alternative point of view of the
strategies for Banach and Mazur, that is equivalent, but more suitable
for mathematical proofs.

The intuition is that we let Adam and Eve play their strategies first
in order to generate a \emph{residual tree}. Then Banach and Mazur
play their game as usual. Formally, given a strategy $\sigma$ for $\e$
and a strategy $\tau$ of $\a$, and a state $s$ we build a subset
$\calG_f(\sigma,\tau)$ of $S(AS)^*$ as follows:
\begin{itemize}
\item the initial state $s$ is in $\calG_f(\sigma,\tau)$
\item if $ws$ is in $\Gcal_f(\sigma,\tau)$ and $s\in S_\e$, and
  $ \sigma(ws)=a$, then for all $s'$ such that $(s,a,s')\in\bm$, we
  have that $wsas'$ is in $\Gcal_f(\sigma,\tau)$
\item if $ws$ is in $\Gcal_f(\sigma,\tau)$ and $s\in S_\a$, and
  $ \tau(ws)=a$, then for all $s'$ such that $(s,a,s')\in\bm$, we have
  that $wsas'$ is in $\Gcal_f(\sigma,\tau)$
\end{itemize}

The residual tree $\Gcal(\sigma,\tau)$ is the set of infinite words obtained as limits of sequences in $\Gcal_f(\sigma,\tau)$.

\begin{definition}
  We say that $\e$ and $\b$ \emph{win strategically} from a state $s$
  if from $s$ there exists a strategy $\sigma$ for $\e$ and such that
  for any strategy $\tau$ of $\a$, Banach wins the Banach-Mazur game
  on the residual tree $\Gcal(\sigma,\tau)$.
\end{definition}


\begin{theorem}
  $\e$ and $\b$ win structurally if and only if they win strategically.
\end{theorem}

From the winning strategy $\bsigma$ for $\b$ it is very easy to
extract the winning strategy in the Banach-Mazur Game. Conversely,
given a winning strategy on the residual tree, we can define several
$\bsigma$, by just choosing as we want in the branches not belonging
to the residual tree.

\subsection{Determinacy and positionality}
If Eve and Banach do not win, it means that for each strategy of Eve there is a winning counterstrategy of Adam. 
A stronger case is when Adam has one strategy that wins against all strategies of Eve.

\begin{definition}
  We say that $\a$ and $\m$ \emph{win strategically} from a state $s$
  if from $s$ there exists a strategy $\tau$ for $\a$ and such that
  for any strategy $\sigma$ of $\e$, Banach wins the Banach-Mazur game
  on the residual tree $\Gcal(\sigma,\tau)$.
\end{definition}

\begin{definition}[determinacy]
  A game is \emph{determined} if from every state $s$, either $\e$ and
  $\b$ win strategically, or $\a$ and $\m$ win strategically. 
\end{definition}

Positional strategies can be seen as a
"pruning" of the arena of the game, by removing all the actions that have
not been chosen.

Formally: given two positional strategies $\sigma$ for $\e$ and $\tau$
of $\a$ we build a directed graph on $S$, that we also call
$\calG(\sigma,\tau)$ as follows. For each state $s$:
\begin{itemize}
\item if $s\in S_\e$, and $ \sigma(s)=a$ then for all $s'$ such that
  $(s,a,s')\in\bm$, we have an edge from $s$ to $s'$
  $\Gcal(\sigma,\tau)$
\item if $s\in S_\a$, and $ \tau(s)=a$ then for all $s'$ such that
  $(s,a,s')\in\bm$, we have an edge from $s$ to $s'$ in
  $\Gcal(\sigma,\tau)$
\end{itemize}


In the next section we are going to show the main technical result of
this paper: that a finite $\bm$ game with parity objective is
positionally determined. This means that either $\e$ and $\b$ win
strategically with a positional strategy for $\e$, or $\a$ and $\m$
win strategically with  a positional strategy for $\a$.

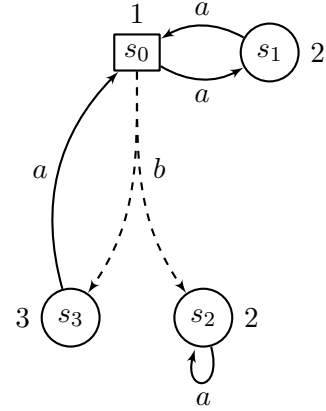
\begin{wrapfigure}{r}{0.4\textwidth}
  \begin{center}
    \begin{tikzpicture}[>=latex', join=bevel, thick]
      \node[] (0) at (0bp, 0bp) [draw, rectangle]{$s_0$}; 
      \node [above] at (0.north) {$1$};
      \node[] (1) at (50bp, 0bp) [draw, circle] {$s_1$};
      \node [right] at (1.east) {$2$};
      \node[] (4) at (25bp, -100bp) [draw, circle] {$s_2$};
      \node [right] at (4.east) {$2$};
      \node[] (5) at (-25bp, -100bp) [draw, circle] {$s_3$};
      \node [left] at (5.west) {$3$};
      \draw[->] (0) [bend right] to node [below] {$a$} (1); 
      \draw[->] (1) [bend right] to node [above] {$a$} (0); 
      \path[->, dashed] (0) edge [out=-90, in=125] node [above right] {$b$}
      (4); 
      \path[->, dashed] (0) edge [out=-90, in=55] node [above right] {} (5);
      \draw[->] (4) edge [loop below] node [below] {$a$} (1);
      \draw[->] (5) [bend left] to node [left] {$a$} (0); 
    \end{tikzpicture}
    \caption{\label{fig:determinacy} A two-player $\bm$ game. The
      relation $\bm$ is defined by the edges of the graph. The dashed
      edges are the one where $\b$ and $\m$ can make important choices}
  \end{center}
\end{wrapfigure}

\subsection{Don't let your right hand know what your left hand is doing}
The reader may wonder if the complex alternation of quantifier in the
definition of structural win is necessary. She may have preferred the
following definition:

We say that $\e$ and $\b$ win from a state $s$ if from $s$ there
exists a strategy $\sigma$ for $\e$ and a strategy $\bsigma$ for $\b$
such that against any strategy $\tau$ of $\a$ and any strategy $\mtau$
for $\m$ the induced pre-play is a play whose flattening is in $\Phi$.

If we considered this definition,
we would give too much power to
$\a$. $\a$ and $\m$ could join forces and win. Indeed, they could take
advantage of the knowledge of $\b$'s strategy, to infer which one of
the Nature players has the lead.  These ideas are presented in the
following example.

\begin{example}
  \label{ex:determinacy}
  Consider the arena depicted in Figure~\ref{fig:determinacy} (which
  is essentially taken from \cite{ACV10}). The only state where there
  is a choice for an Arena player is $s_0$ and we suppose it belongs
  to $\a$. If in this state $\a$ always chooses to go visit $s_1$, he
  will lose. His only chance to win is to move the play in the bottom
  component of the arena. He does this by playing action $b$ when in
  $S_0$ from time to time.  
  if he was playing against a randomised
  nature, then he will almost-surely lose as
  he cannot always avoid state $s_2$.

  However, if we suppose both $\a$ and $m$ know the strategy of $\b$
  i.e. they chose their strategy after $\b$, they can agree on a joint
  strategy as follows: $\a$ chooses to visit $s_1$ as long as $\b$ is
  playing. Once $\b$ passes his turn by playing $\bot$, immediately $\a$
  chooses the action $b$ and then $\m$ helps him by choosing
  $s_3$. Thus, the play would visit $s_3$ infinitely often and $\a$
  would win. Notice that $\a$ needs an infinite memory to implement
  this strategy. In particular, he needs to remember the entire
  history so he can know when does $\b$ passes his turn.

  This shows that with the alternative definition, we cannot have a
  correspondence between the two notions of the game.
\end{example}




\section{Main Theorem}
In this section we prove the main result of the paper i.e.  a transfer
theorem between a stochastic game and the $\bm$ game induced by the following definition.

\begin{definition}
  \label{def:StocToBM}
  Given a two-player stochastic parity games $\calG$, we define
  $\bar{\calG}$ to be the two-player $\bm$ game induced by $\calG$ by
  simply defining the relation $\bm(s,a)$ as $\supp(\calP(s,a))$ for
  any state $s$ and action $a$.
\end{definition}

\subsection{A transfer theorem}

\begin{theorem}[Transfer theorem]\label{thm:trsfrt}
  On a finite arena with parity objective, Eve wins almost-surely the
  stochastic game, if and only if Eve wins the two-player parity $\bm$
  game induced.
\end{theorem}

The proof of the above theorem uses the following theorem, interesting in itself.

\begin{theorem}\label{thm:PosDet}
  Two-player Parity $\bm$ games are positionaly determined.
\end{theorem}

\begin{proof}[Proof of Theorem~\ref{thm:trsfrt}]
  Assume that $\e$ wins from some state $s$ in $\bar{\calG}$. By
  Theorem~\ref{thm:PosDet} there exists a positional winning strategy
  of $\e$. Suppose toward a contradiction that $\e$ does not have an
  almost-surely wining strategy in the Stochastic game $\calG$. By
  positional determinacy, there exists a positional strategy $\tau$
  that is positively winning for $\a$. Let us play this strategy
  against any positional strategy of $\e$.  Let $\sigma$ be a
  positional strategy for $\e$, the pair $(\sigma,\tau)$ induces a
  Markov chain where the objective of $\a$ is satisfied
  positively. Hence, in the residual tree
  $\bar{\calG}(\sigma,\tau)$ we have that $\m$ wins~\cite{VV12}. In
  particular, $\a$ and $\m$ have strategies to win against any
  positional strategy of $\e$ and any strategy of $\b$. Thus, we have
  concluded that all the positional strategies of $\e$ are not
  winning, contradiction.

  Let us prove the other direction. Assume that $\e$ wins
  almost-surely, then by Theorem~\ref{thm:stocPos} she has a positional
  almost-surely winning strategy.
  Suppose toward a contradiction that $\e$ and $\b$ do not win in the
  game $\bar{\calG}$. By positional determinacy, there exists a
  positional winning strategy $\tau$ for $\a$. Let $\sigma$ be a
  positional strategy for $\e$, in the residual tree
  $\bar{\calG}(\sigma,\tau)$, $\m$ wins and because both $\sigma$ and
  $\tau$ are positional, it follows that $\m$ wins in the Markov chain
  induced $\calG(\sigma,\tau)$. Using Theorem~\ref{thm:MCBM}, it follows that
  $\a$ wins positively against any positional strategy of $\e$, a contradiction.
\end{proof}

\subsection{Winning States}

We now turn our attention to the proof Theorem~\ref{thm:PosDet}.  In
order to prove this theorem, we introduce a ew technical tools.

We also introduce the notion of \emph{lead}. Intuitively, during a play of the game, we say that $\m$ (resp. $\b$)
has the lead if we are following the strategy of $\m$
(resp. $\b$). Formally,

\begin{definition}
  Let $\calG(\sigma, \tau)$ be a residual tree induced by the pair
  $(\sigma,\tau)$ and let $\rho$ be a finite play in
  $\calG(\sigma, \tau)$. Then, $\m$ has the lead along $\rho$ if
  $\rho$ contains an even number of occurrences of $\bot$. Otherwise
  $\b$ has the lead.
\end{definition}

An important structural notion is the one of subgames

\begin{definition}[Subgame]
  Let $Q$ be a subset of $S$, $Q$ induces a subgame $\calG(Q)$ if 
  \[
    \forall q \in Q,~\exists a \in A,~\bm(q,a) \subseteq Q
    \enspace.
  \]
\end{definition}

\begin{definition}[Attractor]
  The attractor to $U$ for Eve, denoted $\wattr_{\e}(U,S) \subseteq S$,
  is the limit of the following sequence:
  \begin{align*}
    \wattr_{\e}^0(U,S)=U
    \enspace,
  \end{align*}
  and for any $i\ge 0$
  \begin{align*}
    \wattr_{\e}^{i+1}(U,S) =\wattr_{\e}^i(U,S)&
                                                \cup
                                                \set{s\in S_{\e} \mid \exists a \in A,
                                                \bm(s,a) \cap
                                                \wattr_{\e}^i(U,S) \neq
                                                \emptyset}\\
                                              &\cup
                                                \set{s\in S_{\a} \mid \forall a \in A,
                                                \bm(s,a) \cap
                                                \wattr_{\e}^i(U,S) \ne
                                                \emptyset} \enspace.
  \end{align*}
\end{definition}

The set $\wattr_\a(U,S)$ is defined similarly for $\e$. The states of
$\wattr_\e(U,S)$ enjoy the following property:

\begin{proposition}
  Let $s$ be in $\wattr_\e(U,S)$, there exists a strategy $\sigma$ for
  $\e$ such that against any strategy $\tau$ for $\a$, $\b$ can reach
  $U$ in $\calG(\sigma,\tau)$ if he has the lead from $s$.
\end{proposition}

Obviously, he same claim holds for $\a$ if stated accordingly.

\begin{definition}[Trap]
  \label{lemma:trap}
  A trap for $\a$ is a subset $Q$ of $S$, such that:
  \begin{align*}
    &\forall q \in Q\cap S_\a,~
      \forall a\in A,~\bm(q,a)  \subseteq Q
      \enspace,\\
    &\forall q \in Q\cap S_\e,~
      \exists a\in A,~\bm(q,a)  \subseteq Q
      \enspace.\\
  \end{align*}
\end{definition}
Intuitively, a trap for $\a$, $\calG(Q)$ is a subgame where $\e$ has a
strategy to force the play to never leave $Q$. We define traps for
$\e$ similarly.

\begin{lemma}
  The complement of $\wattr_{\e}(U,S)$ is a trap for $\e$.
\end{lemma}
That is if we let $V$ be the set $S\setminus \wattr_{\e}(U,S)$, then
$\e$ is trapped away from $U$ in $\calG(V)$. We will denote such a
subgame by $\away_\e(U,S)$.

Let $U$ be a subset of $U$ and $(\sigma,\tau)$ a couple of strategies, we define the set $\s(U,S)$
For a subset of states $U$, we define the set 
$\s(U,S)$ as follows:
\[
  \s(U,S)=
  \set{
    s \in S\mid
    \exists \rho \in \paths(\calG(\sigma,\tau)),~
    \first(\rho) = s \land
    \last(\rho) \in \wattr(U,S)
  }
  \enspace,
\]
where $\first(\rho)$ is the first element of the path $\rho$. Intuitively, this is the set of states in $S$ from where $\b$ has a move to reach $U$.
We denote its complement by in $\us(U,S)$

Finally, we define the set $S_d$ as the set of states with priority $d$.

\begin{figure}[h!]
  \begin{minipage}[b]{0.5\textwidth}
    \begin{center}
      \begin{tikzpicture}[>=latex',join=bevel, thick, scale = 1.2]
        \draw[rounded corners] (0,0)--(3,0)--(3,2)--(0,2)--cycle;
        \node (1) at (2.5, .25) [] {\tiny{$S_d$}};
        \draw[dashed] (1.8,0.2).. controls (2,.75)..(3,1);  
        \fill[rounded corners, opacity=.2] (1.6,.4).. controls (2,1)
        and (2.3, 1).. (3,1.3)--(3,0)--(2,0)--cycle;  
        \draw[ rounded corners](0,2)--(3,2)--(3,0)--(2,0)--cycle;
        \node (3) at (1.5,1.5) [] {\tiny{$\away(S_d,S)$}}; 
      \end{tikzpicture}
      \subcaption{\label{fig:even}Largest priority $d$ is even}
    \end{center}
  \end{minipage}
  \begin{minipage}[b]{0.5\textwidth}
    \begin{center}
      \begin{tikzpicture}[>=latex',join=bevel, thick, scale = 1.2]
        \draw[rounded corners] (0,0)--(3,0)--(3,2)--(0,2)--cycle;
        \draw[rounded corners, dashed] (2.5,2)--(2.5,0); 
        \node (2) at (2.7,.8) [] {\tiny{$S_d$}}; 
        \draw[rounded corners] (2.3,2).. controls (2.3,1) ..(1.8,0); 
        \draw[rounded corners, dotted]
        (0,1.6)--(0,1.5)--(.5,1.5)--(.5,2)--(.4,2);
        \node (3) at (.3,1.7) [] {\tiny{$R_1$}}; 
        \draw[rounded corners, dotted] (.4,1.5) .. controls (2.1, 1.3) .. (3,1); 
        \node (4) at (1.3,1.7) [] {\tiny{$Z_1$}}; 
        \draw[rounded corners, dotted] (0,1)--(.5,1)--(.5,1.5); 
        \draw[rounded corners, dotted]
        (.4,1).. controls (2.5, .5) ..(2.9,0)--(2.8,0);
        \node (5) at (.3,1.2) [] {\tiny{$R_2$}}; 
        \node (6) at (1.3,1) []{\tiny{$ Z_2$}}; 
        \fill[opacity=.2, rounded corners]
        (.4,1).. controls (2.5, .5) ..(2.9,0)--(0,0)--(0,1); \node (1) at
        (1.3,.5) [] {$\vdots$};
      \end{tikzpicture}
    \subcaption{\label{fig:odd}Largest priority $d$ is odd}
    \end{center}
  \end{minipage}
  \caption{\label{fig:winReg}Wining region in a two-player $\bm$ parity game}
\end{figure}
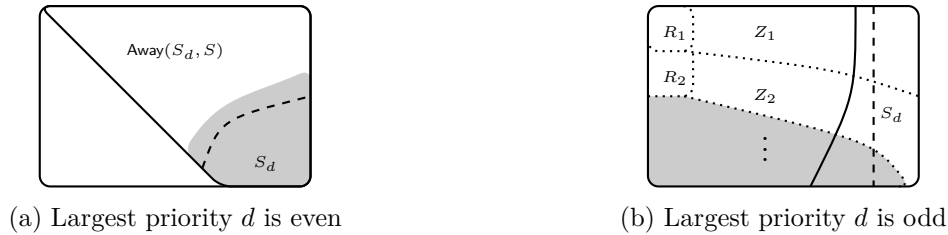

In Figure~\ref{fig:winReg}, we depicted the main ideas behind this construction. In particular we consider two cases: the first one (c.f. Figure~\ref{fig:even}) is when the largest priority $d$ is \emph{even}. In this case we claim that the winning set is given by the largest trap $T$ for $\a$ such that the subgame $\calG(T)$ satisfies the following property; $\e$ wins in $\away(S_d,T)$. Basically, this follows from the fact that in this trap if $\e$ applies the attration strategy when in $\wattr_\e(S_d,T)$ (c.f. gray area in~\ref{fig:even}) and applies her winning strategy in the subgame $\away(S_d,T)$, then we can show that under any strategy $\tau$ of $\a$, $\b$ has a strategy to force any play $\rho$ in $\calG(\sigma,\tau)$ to enter infinitely often  in $\s(S_d,T)$ or to always stay in in $\us(S_d,T)$. This is formalized in the proof of Proposition~\ref{prop:even}.

The second case is the when largest priority $d$ is odd
(c.f. Figure~\ref{fig:odd}). We claim that the winning set of states is given by the largest trap $T$ such that $T$ can be partionned into a sequence of subgames such that each one of them is winning for $\e$. The key argument in this construction is that one can define a total order on the subgames obtained, such that a play can only escape a subgame to visit a subgame that is smaller, thus eventually any play eventually remains forever in a subgame that is winning for $\e$. Details of the correctness are exposed in the proof of Proposition~\ref{prop:odd}.

Let us first describe a procedure to obtain the set winning states for $\e$ in a two-player 
parity $\bm$ game. This is done thanks to Algorithm~\ref{alg:bmpg}. The inner loop that stats in  Line~\ref{line:evenLoop} is the formalization of Figure~\ref{fig:even}; it inductively constructs traps for $\a$ and check that is satisfies the desired condition.
The inner loop that starts in Line~\ref{line:oddLoop} constructs in a iterative manner the sequence of subgames.

  \begin{algorithm}[h!]
    \caption{\label{alg:bmpg} Procedure to compute the set of winning states in a two-player parity $\bm$ game}
    \algsetup{linenodelimiter=}
    \begin{algorithmic}[1] 
      \REQUIRE{Two-player parity $\bm$ game $\calG$ with state space $S$.} 
      \ENSURE{Outputs the winning region for $\e$.}
      \STATE{Let $d$ be the largest priority of $\calG$.}
      \STATE{Let $U_d$ be the set of states with priority $d$}
      \STATE{$U\gets S$}
      \IF{$d$ is even}
        \REPEAT
        \STATE{Compute $\away_\e(U_d, U)$\label{line:evenLoop}}
        \STATE{Compute $R$, the winning region for $\e$ in the subgame $\away(U_d, U)$\label{line:evenInductiveCall}}
        \STATE{$R'\gets U\setminus R$}
        \STATE{Compute $\away_\a(R',U)$\label{line:winAEven}}
        \STATE{$U\gets \away_\a(R',U)$}
        \UNTIL{$R'= \emptyset$}
        \RETURN{$U$}
      \ELSIF{$d$ is odd}
        \STATE{$R'\gets\emptyset$}
        \REPEAT
        \STATE{Compute $\away_\a(U_d,U)$\label{line:oddLoop}}
        \STATE{Compute $R$, the winning region for $\e$ in the 
        subgame $\away(U_d,U)$ \label{line:oddInductiveCall}}
        \STATE{Compute $\wattr_\e(R,U)$, the attractor of $\e$ to $R$ in $\calG(U)$}
        \STATE{$R' \gets R'\cup\wattr_\e(R,U)$}
        \STATE{$U\gets \away_\a(R,U)$}
        \UNTIL{$R = \emptyset$}
       \RETURN{$R'$}
      \ENDIF
    \end{algorithmic} 
  \end{algorithm}

Notice that the above algorithm contains to inductive calls one in Line~\ref{line:evenInductiveCall} and one in Line~\ref{line:oddInductiveCall}. These inductive calls are made on games where the top priority is smaller than $d$, therefore the recursion terminates. Indeed, if $d=0$, $\away_\e(U_d, U)$ is empty, and there is no further recursive call.

\subsection{Correctness of Algorithm~\ref{alg:bmpg}}
We are going to present the most technical part of the paper. The proofs are inspired from the ones presented in \cite{CDGO14} for stochastic games.
While the structure of the proofs are similar, we underline that reasoning in terms of the strategies for Banach and Mazur is more intuitive.

We argue that this is one of the basic contribution of our approach: to concentrate the hard and numeric probabilistic reasoning in one place \cite{VV12}, and then deal more easily with the structural arguments.

\begin{proposition}
  \label{prop:even}
  Let $\calG$ be a parity game where the largest priority $d$ is
  even. All the states 
  in $\calG$ are winning for $\e$ if and only if all the states in $\away(S_d,S)$ are winning for $\e$ in 
  $\calG(\away(S_d,S))$. 
\end{proposition}

\begin{proof}
  Denote by $X$ the set $\wattr_\e(S_d,S)$ and by $Y$ the set
  $\away(S_d,S)$.  
  Assume that $Y$ is winning for $\e$ then $S$ is
  winning for $\e$. 
  Let $\sigma_X$ the strategy induced by the
  attractor set $X$, and $\sigma_Y$ the winning strategy for $\e$ over
  $\calG[Y]$. 
  We define the positional strategy $\sigma$ as follows
  \begin{align*}
    \sigma : S_\e &\to A\\
    s&\mapsto
       \begin{cases}
         \sigma_X(s) \text{ if } s \in X \enspace,\\
         \sigma_Y(s) \text{ if } s \in Y \enspace.
       \end{cases}
  \end{align*}

  We show now that $\sigma$ is winning.
  Let $\tau$ be an arbitrary strategy for $\a$, and consider the
  (potentially infinite) residual tree $\calG(\sigma,\tau)$. We will show that in
  $\calG(\sigma,\tau)$, $\b$ wins. 

  Define the strategy $\bsigma$ for $\b$ as follows:
  \begin{align*}
    \bar{\sigma} : \calG(\sigma,\tau) &\to S \cup \bot\\
    \rho&\mapsto
          \begin{cases}
            \rho'\bot \text{ s.t. } \last(\rho') \in S_d \text{ if }
            \last(\rho) \in \s(S_d,S) \enspace,\\
            \bsigma_Y(\bar{\rho}) \text{ if } \last(\rho) \in \us(S_d,S) \enspace,
          \end{cases}
  \end{align*}
  where $\bar{\rho}$ is the longest suffix of $\rho$ that only contains
  states from $\us(S_d,S)$.

  Notice that $\us(S_d, S)$ is a subset of $Y$ and that any play that
  starts in $\us(S_d, S)$, remains in $\us(S_d, S)$.
 
  To see that $\bsigma$ is winning, let $\rho$ be a finite play in
  $\calG(\sigma,\tau)$ that respects $\bsigma$.  and assume that
  $\last(\rho)$ is in $\us(S_d,S)$, then the subsequent play from
  $\bar{\rho}$ is winning because it respects $\bsigma_Y$ which
  is winning, the fact that $\parity$ is prefix independent entails
  that the subsequent play from $\rho$ is winning as well.

  Now assume that $\last(\rho)$ is in $\s(S_d,S)$, then
  \begin{itemize}
  \item if it is $\b$'s turn, by playing according to $\bsigma$ he
    will visit a state with priority $d$ and passes the lead in a
    state in $S_d$, thus visiting a state with priority $d$.
  \item If it is $\m$'s turn, he either plays and passes the lead
    in $\s(S_d,S)$, in which case $\b$ can visit $S_d$ again, or he
    passes the lead in $\us(S_d,S)$, in which case the previous case
    applies.
  \end{itemize}
  Finally, notice that $\sigma$ is positional, thus $\e$ wins using a
  positional strategy.

  Let us prove the converse, assume that all states in $\calG$ are
  winning. By Lemma~\ref{lemma:trap} we know that $Y$ is a trap for $\e$, thus
  if $\e$ does not win in $Y$ it cannot win in $\calG$.
\end{proof}

\begin{proposition}
  \label{prop:odd}
  Let $\calG$ be a parity game where the largest priority $d$ is
  odd. All the states in $\calG$ are winning for $\e$. if and only if
  there exists a partition $\set{Z_i}_{1\le i \le k}$ of $S$ and non
  empty sets $R_i, U_i$ for $i = 1,\ldots, k,$ such that $U_1 = S$ and
  for all $1\le i\le k$
  \begin{itemize}
  \item[1)] $R_i \subseteq U_i\setminus {(U_{i})}_{d}$ is a trap for $\a$ in
    $\calG[U_i]$ and all $R_i$ are winning in $\calG[U_i]$;
  \item[2)] $Z_i = \wattr_\e(R_i, U_i)$;
  \item[3)] $U_{i+1} = U_i \setminus Z_i$.
  \end{itemize}
\end{proposition}

\begin{proof}
  Let $\sigma_i$ be the winning strategy for $\e$ in $\calG(R_i)$, let
  also $\asigma_i$ be the positional strategy induced by the attractor
  $Z_i$.  Define the following strategy $\sigma$ for $\e$,
  \begin{align*}
    \sigma : S_\e &\to A\\
    s&\mapsto
       \begin{cases}
         \sigma_i(s) \text{ if } s \in R_i \enspace,\\
         \asigma_i(s) \text{ if } s \in Z_i \enspace.
       \end{cases}
  \end{align*}  
  
  Let $\tau$ be an arbitrary strategy for $\a$, and let us show that
  $\b$ wins in the (potentially infinite) residual tree
  $\calG(\sigma,\tau)$.  We define the following strategy $\bsigma$
  for $\b$:
  \begin{align*}
    \bar{\sigma} : \calG(\sigma,\tau) &\to S \cup \bot\\
    \rho&\mapsto
          \begin{cases}
            \rho'\bot \text{ s.t. } \last(\rho') \in R_i \text{ if }
            \last(\rho) \in Z_i \enspace,\\
            \bsigma_i(\bar{\rho}) \text{ if } \last(\rho) \in R_i
            \enspace,
          \end{cases}
  \end{align*}
  where $\bar{\rho}$ is the longest suffix of $\rho$ that only
  contains states from $R_i$.  
  
  Let us show that $\bsigma$ is winning in the residual tree
  $\calG(\sigma,\tau)$.
  Let $\rho$ be a finite play that respects $\bsigma$, Since $Z_i$ is
  a partition, it follows that $\last(\rho)$ is in some $Z_i$.
  Assume that it is $\m$'s turn, then whatever action he plays, he
  will either $i)$ pass the lead  in a state in $Z_i$ or $ii)$ pass
  the lead in some $Z_j$ such that $i \neq j$. 

  If $i)$ holds, then $\b$ will move the play to $R_i$ and all the
  subsequent plays will remain there since $R_i$ is a trap for
  $\a$. Moreover, since $R_i$ is winning for $\e$ over $\calG(U_i)$ it
  follows that $\b$ wins.

  If $ii)$ holds, then since there are only finitely many such $j$ and
  because by construction we have $i < j \le k$ ,
  the play will settle in some $Z_j$ and $\b$ can move the play to
  $R_j$ and the previous arguments apply.

  The converse implication now. Assume that $\e$ wins from every state
  in $\calG$, we construct a partition of $S$ that meets the
  requirements of the statment. Let $X$ be the set
  $\wattr_\a(S_d,S)$ and let $Y$ be $S\setminus X$. We claim that since all states in $\calG$ are
  winning for $\e$, then the winning region of $\e$ in $\calG(Y)$ is
  non-empty. Assume toward a contradiction that it is not the case,
  then it means that all states in $Y$ are loosing
  c.f. whenever $\e$ chooses a strategy, there exists a strategy
  $\tau$ in $\calG(Y)$ such that in the residual tree $\calG(Y)$, $\b$
  does not win). Let us show that this implies that all states in $S$
  are loosing for $\e$. Consider the set of states $\s(S_d,S)$ and
  $\us(S_d,S)$ as defined in the proof of Proposition~\ref{prop:even}. As long
  as the current play is in $\s(S_d,S)$ then $\m$ can visit states in
  $S_d$, if the play moves to $\us(S_d,S)$ then $\m$ wins, a
  contradiction. Thus the winning region in $\calG[Y]$ for $\e$ is
  non-empty (it is also a trap for $\a$), let this region be $R_1$, and $Z_1$ be
  $\wattr_\e(R_1,S)$. If $S\setminus Z_1$ is empty, then we are
  done. Otherwise we repeat the construction over $S \setminus
  Z_1$. Since $Z_1$ is non empty, this construction terminates.
\end{proof}

\subsection{Posional strategies}

\begin{proposition}
  \label{prop:Epos}
  If $\e$ wins, then she has a positional winning strategy
\end{proposition}

\begin{proof}
  We prove this by induction over the largest priority $d$ available in the game.
  
  If $d = 0$, then any positional strategy is winning.
  
  Assume now that $d$ is even, then from the proof of Proposition~\ref{prop:even}, we now that the winning strategy $\sigma$ uses two strategies; the attraction strategy $\sigma_X$ which is positional, and $\sigma_Y$ defined in a subgame with less priorities thus it is positional by induction.
  
  If $d$ is odd, then thanks to the proof of Proposition~\ref{prop:odd} we know that each $\asigma_i$ is positional because it is induced from an attractor, and each
  $\sigma_i$ is positional again by induction.
\end{proof}
 
\begin{proposition}
  \label{prop:Apos}
  If $\a$ wins, then he has a positional winning strategy
\end{proposition}

\begin{proof}
  By induction on the number of states. If $|S| = 1$ is then the result follows.
  
  Assume it is the case for any game with state space $S$
  We will again consider two cases; when the highest priority is even and when it is odd.
  
  In former case, thanks to Algorithm~\ref{alg:bmpg}-Line~\ref{line:winAEven}, we know that the winning region of $\a$ is a finite union of sets of the form $\wattr_\a(R',U)$ where $\a$ wins the parity game played in $R'$, thus since $\m$ has the lead in the initial state, it suffices for $\a$ to play a positional strategy to reach with the help of $\m$, $R'$ is reached in one move. Once in $R'$ $\a$ applies a positional winning strategy that exists by induction. Any play that respects this strategy will stay forever in $R'$ since it is a trap for $\e$.
  
  In the latter case, from Algorithm~\ref{alg:bmpg}-Line~\ref{line:oddLoop}, we know that the winning region is a finite union of sets of the form $\wattr_\a(U_d,U)$ such that $i)$ $U$ induces a subgame and $ii)$ $d$ is the largest priority in $U$. Thus, $\a$ can always apply an attraction strategy and with the help of $\m$ continuously visit states in $U_d$. This strategy is clearly positional.
\end{proof}

These two last propositions yield Theorem~\ref{thm:PosDet}.

\vspace{12pt}

We just mention here that Theorem~\ref{thm:trsfrt} can be easily extended as usual from parity objectives to $\omega$-regular objectives, by just making the product of the game and the parity automaton that accepts the objective. We omit the straightforward details here.


\section{Conclusions}
What we have shown in this paper, while technically non trivial, can
still be considered only as a "sanity" check - the Banach-Mazur game
can replace the probabilities in a suitable setting. While the structure of the proofs mimics the corresponding probabilistic proofs, the advantage is that all probabilistic reasoning is formally "factored out", so that following the proof becomes easier.

However, we consider that our result can be made to go
further. In particular, we note that in this paper, we use the fact
that stochastic parity games are positionaly determined. This result
is needed in our proof of Theorem~\ref{thm:trsfrt}. What we would like
to achieve is to use positional determinacy of EBAM games, to
\emph{prove} positional determinacy of stochastic games. This what we
set out to do next.

Another contribution of this paper is to simplify the presentation of \cite{ACV10}. Notably there the authors introduced the heavy concept of "move tree" which we have simplified here.



\nocite{*}
\bibliographystyle{eptcs}
\bibliography{references}

\begin{thebibliography}{1}
\providecommand{\bibitemdeclare}[2]{}
\providecommand{\surnamestart}{}
\providecommand{\surnameend}{}
\providecommand{\urlprefix}{Available at }
\providecommand{\url}[1]{\texttt{#1}}
\providecommand{\href}[2]{\texttt{#2}}
\providecommand{\urlalt}[2]{\href{#1}{#2}}
\providecommand{\doi}[1]{doi:\urlalt{http://dx.doi.org/#1}{#1}}
\providecommand{\bibinfo}[2]{#2}

\bibitemdeclare{inproceedings}{ACV10}
\bibitem{ACV10}
\bibinfo{author}{Eugene \surnamestart Asarin\surnameend},
  \bibinfo{author}{Rapha{\"{e}}l \surnamestart Chane{-}Yack{-}Fa\surnameend} \&
  \bibinfo{author}{Daniele \surnamestart Varacca\surnameend}
  (\bibinfo{year}{2010}): \emph{\bibinfo{title}{Fair Adversaries and
  Randomization in Two-Player Games}}.
\newblock In: {\sl \bibinfo{booktitle}{Foundations of Software Science and
  Computational Structures, 13th International Conference, {FOSSACS} 2010, Held
  as Part of the Joint European Conferences on Theory and Practice of Software,
  {ETAPS} 2010, Paphos, Cyprus, March 20-28, 2010. Proceedings}}, pp.
  \bibinfo{pages}{64--78}, \doi{10.1007/978-3-642-12032-9_6}.
\newblock \urlprefix\url{https://doi.org/10.1007/978-3-642-12032-9_6}.

\bibitemdeclare{book}{BK08}
\bibitem{BK08}
\bibinfo{author}{Christel \surnamestart Baier\surnameend} \&
  \bibinfo{author}{Joost{-}Pieter \surnamestart Katoen\surnameend}
  (\bibinfo{year}{2008}): \emph{\bibinfo{title}{Principles of model checking}}.
\newblock \bibinfo{publisher}{{MIT} Press}.

\bibitemdeclare{article}{BHM15}
\bibitem{BHM15}
\bibinfo{author}{Thomas \surnamestart Brihaye\surnameend},
  \bibinfo{author}{Axel \surnamestart Haddad\surnameend} \&
  \bibinfo{author}{Quentin \surnamestart Menet\surnameend}
  (\bibinfo{year}{2015}): \emph{\bibinfo{title}{Simple strategies for
  Banach-Mazur games and sets of probability 1}}.
\newblock {\sl \bibinfo{journal}{Inf. Comput.}} \bibinfo{volume}{245}, pp.
  \bibinfo{pages}{17--35}, \doi{10.1016/j.ic.2015.06.004}.
\newblock \urlprefix\url{https://doi.org/10.1016/j.ic.2015.06.004}.

\bibitemdeclare{inproceedings}{CJH03}
\bibitem{CJH03}
\bibinfo{author}{Krishnendu \surnamestart Chatterjee\surnameend},
  \bibinfo{author}{Marcin \surnamestart Jurdzinski\surnameend} \&
  \bibinfo{author}{Thomas~A. \surnamestart Henzinger\surnameend}
  (\bibinfo{year}{2003}): \emph{\bibinfo{title}{Simple Stochastic Parity
  Games}}.
\newblock In: {\sl \bibinfo{booktitle}{Computer Science Logic, 17th
  International Workshop, {CSL} 2003, 12th Annual Conference of the EACSL, and
  8th Kurt G{\"{o}}del Colloquium, {KGC} 2003, Vienna, Austria, August 25-30,
  2003, Proceedings}}, pp. \bibinfo{pages}{100--113},
  \doi{10.1007/978-3-540-45220-1_11}.
\newblock \urlprefix\url{https://doi.org/10.1007/978-3-540-45220-1_11}.

\bibitemdeclare{inproceedings}{GH10}
\bibitem{GH10}
\bibinfo{author}{Hugo \surnamestart Gimbert\surnameend} \&
  \bibinfo{author}{Florian \surnamestart Horn\surnameend}
  (\bibinfo{year}{2010}): \emph{\bibinfo{title}{Solving Simple Stochastic Tail
  Games}}.
\newblock In: {\sl \bibinfo{booktitle}{Proceedings of the Twenty-First Annual
  {ACM-SIAM} Symposium on Discrete Algorithms, {SODA} 2010, Austin, Texas, USA,
  January 17-19, 2010}}, pp. \bibinfo{pages}{847--862},
  \doi{10.1137/1.9781611973075.69}.
\newblock \urlprefix\url{https://doi.org/10.1137/1.9781611973075.69}.

\bibitemdeclare{book}{Ox}
\bibitem{Ox}
\bibinfo{author}{J.C. \surnamestart Oxtoby\surnameend} (\bibinfo{year}{2014}):
  \emph{\bibinfo{title}{Measure and Category}}.
\newblock \bibinfo{publisher}{Springer}.

\bibitemdeclare{inproceedings}{Staiger97}
\bibitem{Staiger97}
\bibinfo{author}{Ludwig \surnamestart Staiger\surnameend}
  (\bibinfo{year}{1997}): \emph{\bibinfo{title}{Rich omega-Words and Monadic
  Second-Order Arithmetic}}.
\newblock In: {\sl \bibinfo{booktitle}{Computer Science Logic, 11th
  International Workshop, {CSL} '97, Annual Conference of the EACSL, Aarhus,
  Denmark, August 23-29, 1997, Selected Papers}}, pp.
  \bibinfo{pages}{478--490}, \doi{10.1007/BFb0028032}.
\newblock \urlprefix\url{https://doi.org/10.1007/BFb0028032}.

\bibitemdeclare{article}{VV12}
\bibitem{VV12}
\bibinfo{author}{Hagen \surnamestart V{\"{o}}lzer\surnameend} \&
  \bibinfo{author}{Daniele \surnamestart Varacca\surnameend}
  (\bibinfo{year}{2012}): \emph{\bibinfo{title}{Defining Fairness in Reactive
  and Concurrent Systems}}.
\newblock {\sl \bibinfo{journal}{J. {ACM}}}
  \bibinfo{volume}{59}(\bibinfo{number}{3}), pp. \bibinfo{pages}{13:1--13:37},
  \doi{10.1145/2220357.2220360}.
\newblock \urlprefix\url{http://doi.acm.org/10.1145/2220357.2220360}.

\bibitemdeclare{inproceedings}{Zielonka04}
\bibitem{Zielonka04}
\bibinfo{author}{Wieslaw \surnamestart Zielonka\surnameend}
  (\bibinfo{year}{2004}): \emph{\bibinfo{title}{Perfect-Information Stochastic
  Parity Games}}.
\newblock In: {\sl \bibinfo{booktitle}{Foundations of Software Science and
  Computation Structures, 7th International Conference, {FOSSACS} 2004, Held as
  Part of the Joint European Conferences on Theory and Practice of Software,
  {ETAPS} 2004, Barcelona, Spain, March 29 - April 2, 2004, Proceedings}}, pp.
  \bibinfo{pages}{499--513}, \doi{10.1007/978-3-540-24727-2_35}.
\newblock \urlprefix\url{https://doi.org/10.1007/978-3-540-24727-2_35}.

\end{thebibliography}
\end{document}
